\documentclass[runningheads]{llncs}

\usepackage{latexsym,amsmath,amssymb}
\usepackage{graphicx,epsfig}
\usepackage{algo}
\usepackage{url}
\usepackage{multirow}
\usepackage[nocompress]{cite}

\newcommand{\Suff}{\textit{Suff}}
\newcommand{\Pref}{\textit{Pref}}
\newcommand{\Fact}{\textit{Fact}}
\newcommand{\PL}{P\!L}

\def\P{p}

\newcommand{\fmax}{F}
\newcommand{\pr}{P}
\newcommand{\pos}{\mathit{pos}}

\newcommand{\fmin}{f}

\newcommand{\select}{{\textit{select}}}
\newcommand{\rank}{{\textit{rank}}}

\renewcommand{\epsilon}{\varepsilon}

\newcommand{\PNF}{\mathrm{PNF}}

\spnewtheorem{observation}{Observation}{\bfseries}{\itshape}

\begin{document}

\title{On Prefix Normal Words}

\author{Gabriele Fici \inst{1} \and Zsuzsanna Lipt\'ak \inst{2}}

\institute{I3S, CNRS \& Universit\'e de Nice-Sophia Antipolis, France \\ \email{fici@i3s.unice.fr} \and AG Genominformatik, Technische Fakult\"at, Bielefeld University, Germany\\ \email{zsuzsa@cebitec.uni-bielefeld.de}}

\maketitle

\begin{abstract}
We present a new class of binary words: the prefix normal words. They are defined by the property that for any given length $k$, no factor of length $k$ has more $a$'s than the prefix of the same length. These words arise in the context of indexing for jumbled pattern matching (a.k.a.\ permutation matching or Parikh vector matching), where the aim is to decide whether a string has a factor with a given multiplicity of characters, i.e., with a given Parikh vector. Using prefix normal words, we give the first non-trivial characterization of binary words having the same set of Parikh vectors of factors. We prove that the language of prefix normal words is not context-free and is strictly contained in the language of pre-necklaces, which are prefixes of powers of Lyndon words. We discuss further properties and state open problems.
\end{abstract}


\section{Introduction}

Given a finite word $w$ over a finite ordered alphabet $\Sigma$, the Parikh vector of $w$ is defined as the vector of multiplicities of the characters in $w$. In recent years, Parikh vectors have been increasingly studied, in particular Parikh vectors of factors (substrings) of words, motivated by applications in computational biology, e.g.\ mass spectrometry~\cite{EresLP04,CieliebakELSW04,Boecker07,ADMOP10}. Among the new problems introduced in this context is that of {\em jumbled pattern matching} (a.k.a.\ permutation matching or Parikh vector matching), whose decision variant is the task of deciding whether a given word $w$ (the text) has a factor with a given Parikh vector (the pattern). In~\cite{CiFiLi09}, Cicalese {\em et al.} showed that in order to answer decision queries for binary words, it suffices to know, for each $k$, the maximum and minimum number of $a$'s in a factor of length $k$. Thus it is possible to create an index of size $O(n)$ of a text of length $n$, which contains, for every $k$, the maximum and minimum number of $a$'s in a factor of length $k$, and which allows answering decision queries in constant time.

In this paper, we introduce a new class of binary words, {\em prefix normal words}. They are defined by the property that for any given length $k$, no factor of length $k$ appearing in the word has more $a$'s than the prefix of the word of the same length. For example, the word $aabbaaba$ is not a prefix normal word, because the factor $aaba$ has more $a$'s then the prefix of the same length, $aabb$. 

We show that for every binary word $w$, there is a prefix normal word $w'$ such that,  for every $0\leq k \leq |w|$, the maximum number of $a$'s in a factor of length $k$ coincide for $w$ and $w'$ (Theorem \ref{thm:pnf}). We refer to $w'$ as the {\em prefix normal form} of $w$ (with respect to $a$).

Given a word $w$, a {\em factor Parikh vector} of $w$ is the Parikh vector of a factor of $w$. An interesting characterization of words with the same multi-set of factor Parikh vectors was given recently by Acharya {\em et al.}~\cite{ADMOP10}.
In this paper, we give the first non-trivial characterization of the {\em set} of factor Parikh vectors, by showing that two words have the same set of factor Parikh vectors if and only if their prefix normal forms, with respect to $a$ and $b$, both coincide (Theorem \ref{thm:charparset}). 

We explore the language of prefix normal words and its connection to other known languages. Among other things, we show that this language is not context-free (Theorem \ref{teor:CF}) by adapting a proof of Berstel and Boasson~\cite{BeBo97} for Lyndon words, and that it is properly included in the language of pre-necklaces, the prefixes of powers of Lyndon words (Theorem \ref{theor:PL}). 
We close with a number of open problems.

\medskip

\noindent
{\bf Connection to Indexed Jumbled Pattern Matching.} The current fastest algorithms for computing an index for the binary jumbled pattern matching problem were concurrently and independently developed by Burcsi {\em et al.}~\cite{BuCiFiLi10a} and Moosa and Rahman~\cite{MoRa10}. In order to compute an index of a text of length $n$, both used a reduction to min-plus convolution, for which the current best algorithms have a runtime of $O(n^2/\log n)$. Very recently, Moosa and Rahman~\cite{MoRaJDA} introduced an algorithm with runtime $O(n^2/\log^2 n)$ which uses word-RAM operations. Our characterization of the set of factor Parikh vectors in terms of prefix normal forms yields a new approach to the problem of indexed jumbled pattern matching: Given the prefix normal forms of a word $w$, the index for the jumbled pattern matching problem can be computed in linear time $O(n)$. This implies that any algorithm for computing the prefix normal form with runtime $o(n^2/\log^2 n)$ will result in an improvement for the indexing problem for binary jumbled pattern matching. Since on the other hand,  the prefix normal forms can be computed from the index in $O(n)$ time, we also have that a lower bound for the computation of the prefix normal form would  yield a lower bound for the binary jumbled pattern matching problem.

\section{The Prefix Normal Form}\label{sec:pnf}

We fix the ordered alphabet $\Sigma=\{a,b\}$, with $a<b$. A word $w=w_1\cdots w_n$ over $\Sigma$ is a finite sequence of elements from $\Sigma$. Its length $n$ is denoted by $|w|$. We denote the empty word by $\epsilon$. For any $1\leq i\leq |w|$, the $i$-th symbol of a word $w$ is denoted by $w_{i}$. As is standard, we denote by $\Sigma^n$ the words over $\Sigma$ of length $n$, and by $\Sigma^{*} = \cup_{n\geq 0} \Sigma^n$ the set of finite words  over $\Sigma$. Let $w\in \Sigma^{*}$. If $w=uv$ for some $u,v\in\Sigma^{*}$, we say that $u$ is a \emph{prefix} of $w$ and $v$ is a \emph{suffix} of $w$. A \emph{factor} of $w$ is a prefix of a suffix of $w$ (or, equivalently, a suffix of a prefix). We denote by $\Pref(w)$, $\Suff(w)$, $\Fact(w)$  the set of prefixes, suffixes, and factors of the word $w$, respectively.

For a letter $a\in \Sigma$, we denote by $|w|_{a}$ the number of occurrences of $a$ in the word $w$. The {\em Parikh vector} of a word $w$ over $\Sigma$ is defined as $\P(w)=(|w|_{a},|w|_{b})$. The {\em Parikh set} of $w$ is $\Pi(w)=\{\P(v) \mid  v\in \Fact(w)\}$, the set of Parikh vectors of the factors of $w$. 

Finally, given a word $w$ over $\Sigma$ and a letter $a\in \Sigma$, we denote by $\pr_{a}(w,i) = |w_1\cdots w_i|_a$, the number of $a$'s in the prefix of length $i$ of $w$, and by $\pos_a(w,i)$ the position of the $i$'th $a$ in $w$, i.e., $\pos_a(w,i) = \min\{ k : |w_1\cdots w_k|_a = i\}$. When $w$ is clear from the context, we also write $\pr_a(i)$ and $\pos_a(i)$. Note that in the context of succint indexing, these functions are frequently called $\rank$ and $\select$, cf.~\cite{NavMaek07}: We have $ \pr_a(w,i) = \rank_a(w,i)$ and $ \pos_a(w,i) = \select_a(w,i)$. 

\begin{definition}
Let $w\in \Sigma^*$. We define, for each $0\leq k\leq |w|$, 
\begin{eqnarray*}
\fmax_a(w,k) =  \max\{|v|_a \mid v\in \Fact(w)\cap \Sigma^k\},
\end{eqnarray*}

\noindent the maximum number of $a$'s in a factor of $w$ of length $k$. 
When no confusion can arise, we also write $F_a(k)$ for $F_a(w,k)$. 
The function $\fmax_b(w)$ is defined analogously by taking $b$ in place of $a$. 
\end{definition}

\begin{example}
 Take $w = ababbaabaabbbaaabbab$. In Table~\ref{tab:val}, we give the values of $\fmax_{a}$ and $\fmax_{b}$ for $w$.

\begin{table}
\centering  
\begin{small}
\begin{raggedright}
\begin{tabular}{c *{30}{@{\hspace{2.1mm}}l}}
 $k$  &  0\hspace{1ex}  & 1\hspace{1ex} & 2\hspace{1ex} & 3\hspace{1ex} & 4\hspace{1ex} & 5\hspace{1ex} & 6\hspace{1ex} & 7\hspace{1ex} &
8\hspace{1ex} & 9\hspace{1ex} & 10 & 11 & 12 & 13 & 14 & 15 & 16 & 17 & 18 & 19 & 20 \\
\hline \rule[-6pt]{0pt}{22pt}
$ \fmax_{a}$ &  0    & 1& 2& 3& 3& 4& 4& 4& 5& 5& 6& 7& 7& 7& 8& 8& 9& 9& 9& 10& 10  \\
\hline \rule[-6pt]{0pt}{22pt}
$ \fmax_{b}$ &  0    & 1& 2& 3& 3& 3& 4& 4& 5& 5& 6& 6& 7& 7& 7& 8& 8& 9& 9& 10& 10\\
\hline \\
\end{tabular}
\end{raggedright}\caption{\label{tab:val}The sequences $ \fmax_{a}$ and $ \fmax_{b}$ for the word $w = ababbaabaabbbaaabbab$.}
\end{small}
\end{table}
\end{example}

\begin{lemma}\label{lemma:Fa}
Let $w\in \Sigma^*$. The function $F_a(\cdot) = F_a(w,\cdot)$  has the following property: 

\[ F_a(j) - F_a(i) \leq F_a(j-i) \qquad \text{for all } 0\leq i \leq j \leq |w|. \]
\end{lemma}

\begin{proof}
Assume otherwise. Then there are indices $i\leq j$ such that $F_a(j) - F_a(i) > F_a(j-i)$. Let $v\in \Fact(w)$ be a word that realizes $F_a(j)$, i.e.,\ $|v| = j$ and $|v|_a = F_a(j)$. Let us write $v = v_1v_2\cdots v_j$. Then for the word $u = v_{i+1} \cdots v_j$, we have $|u|_a = |v|_a - |v_1\cdots v_i|_a \geq F_a(j) - F_a(i) > F_a(j-i)$, in contradiction to the definition of $F_a$, since  $|u| = j-i$. \qed 
\end{proof}

We are now ready to show that for every word $w$ there is a word $w'$ which realizes the function $\fmax_a(w)$ as its prefix function $\pr_a(w')$. 

\begin{theorem}\label{thm:pnf}
Let $w\in \Sigma^*$. Then there exists a unique word $w'$ s.t.\ for all $0\leq k\leq |w|$, $\fmax_a(w,k) = \fmax_a(w',k) = \pr_a(w',k)$. We call this word $w'$ the {\em prefix normal form} of $w$ (with respect to $a$), and denote it $\PNF_a(w)$. Analogously, there exists a unique word $w''$,  such that for all $0\leq k \leq |w|$, $\fmax_b(w,k) = \fmax_b(w'',k) = \pr_b(w'',k)$, the prefix normal form of $w$ with respect to $b$, denoted $\PNF_b(w)$.
\end{theorem}

\begin{proof}
We only give the proof for $w'$. The construction of $w''$ is analogous. 
It is easy to see that for $1\le k\le |w|$, one has either $ \fmax_a(w,k)= \fmax_a(w,k-1)$ or $ \fmax_a(w,k)=1+ \fmax_a(w,k-1)$.
Now define the word $w'$ by
 
 $$w'_k = \begin{cases}
a & \quad \text{if $ \fmax_a(w,k)=1+ \fmax_a(w,k-1)$}\\
b & \quad \text{if $ \fmax_a(w,k)= \fmax_a(w,k-1)$}
\end{cases}$$
for every $1\le k\le |w|$.

By construction, we have $\pr_{a}(w',k) = \fmax_{a}(w,k)$ for every $1\le k\le |w|$. We still need to show that $\pr_a(w',k) = \fmax_a(w',k)$ for all $k$, i.e.,\ that $w'$ is in prefix normal form. By definition, $\pr_a(w',k) \leq \fmax_a(w',k)$ for all $k$. Now let $v\in Fact(w')$, $|v|=k$, and $v = w_{i+1}\cdots w_j$. Then $|v|_a = \pr_a(w',j) - \pr_a(w',i) = \fmax_a(w,j) - \fmax_a(w,i) \leq \fmax_a(w,j-i) = \pr_a(w',j-i) = \pr_a(w',k)$, where the inequality holds by Lemma~\ref{lemma:Fa}. We have thus proved that $ \fmax_a(w',k) \leq \pr_a(w',k)$, and we are done.\qed
\end{proof}

\begin{example}
 Let  $w = ababbaabaabbbaaabbab$.  The prefix normal forms of $w$ are the words $$\PNF_{a}(w)=aaababbabaabbababbab,$$ and $$\PNF_{b}(w)=bbbaababababaabababa.$$
 \end{example}

The operators $\PNF_a$ and $\PNF_b$ are idempotent operators, that is, if $u = \PNF_x(w)$ then $\PNF_x(u) = u$, for any $x\in \Sigma$. Also, for any $w\in \Sigma^{*}$ and $x\in \Sigma$, it holds that $\PNF_{x}(w)=\PNF_{x}(\tilde{w})$, where $\tilde{w}=w_{n}w_{n-1}\cdots w_{1}$ is the reversal of $w$. 

The prefix normal forms of a word allow one to determine the Parikh vectors of the factors of the word, as we will show in Theorem \ref{thm:charparset}. We first recall the following lemma from~\cite{CiFiLi09}, where we say that a Parikh vector $q$ {\em occurs} in a word $w$ if $w$ has a factor $v$ with $p(v) = q$.

\begin{lemma}[Interval Lemma, Cicalese {et al.}~\cite{CiFiLi09}]
\label{lemma:continuous}
Let $w\in \Sigma^*$. Fix $1\leq k\leq |w|$. 
If the Parikh vectors $(x_1,k-x_1)$ and $(x_2,k-x_2)$ both occur in $w$, then so does $(y,k-y)$ for any $x_1\leq y \leq x_2$. 
\end{lemma}

The lemma can be proved with a simple sliding window argument. 

\begin{theorem}\label{thm:charparset}
Let $w,w'$ be words over $\Sigma$. Then $\Pi(w) = \Pi(w')$ if and only if $\PNF_a(w) = \PNF_a(w')$ and $\PNF_b(w) = \PNF_b(w')$. 
\end{theorem}

\begin{proof}
Let $\fmin_a(w,k)$ denote the minimum number of $a$'s in a factor of $w$ of
length $k$.
As a direct consequence of Lemma~\ref{lemma:continuous}, we have that
for a Parikh vector $q=(x,y)$, $q\in \Pi(w)$ if and only if
$\fmin_a(w,x+y) \leq x \leq \fmax_a(w,x+y)$. Thus  for two words
$w,w'$, we have $\Pi(w) = \Pi(w')$ if and only if $\fmax_a(w) =
\fmax_a(w')$ and $\fmin_a(w) = \fmin_a(w')$.  It is easy to see that
for all $k$, $\fmin_a(w,k) = k - \fmax_b(w,k)$, thus the last
statement is equivalent to $\fmax_a(w) = \fmax_a(w')$ and $\fmax_b(w)
= \fmax_b(w')$. This holds if and only if $\PNF_a(w) = \PNF_a(w')$ and
$\PNF_b(w) = \PNF_b(w')$, and the claim is proved. \qed
\end{proof}

There is a simple geometrical construction for computing the prefix normal forms of a word $w$, and hence, by Theorem \ref{thm:charparset}, the set $\Pi(w)$ of Parikh vectors occurring in $w$. An example of this construction is given in Fig. \ref{fig:esempio}.

Draw in the Euclidean plane the word $w$ by linking, for every $0\le i \le |w|$, the points $(i,j)$, where $j$ is the difference between the number of $a$'s and the number of $b$'s in the prefix of $w$ of length $i$. That is, draw $w$ by representing each letter $a$ by an upper unit diagonal and each letter $b$ by a lower unit diagonal, starting from the origin $(0,0)$.

Then draw all the suffixes of $w$ in the same way, always starting from the origin. The region of the plane so delineated is in fact the region of points $(x,y)$ such that there exists a factor $v$ of $w$ such that $x=|v|=|v|_{a}+|v|_{b}$ and $y=|v|_{a}-|v|_{b}$. Hence $(|v|_{a},|v|_{b})=(\frac{x+y}{2},\frac{x-y}{2})=\P(v)$ belongs to $\Pi(w)$. 

The region is connected by Lemma \ref{lemma:continuous}, in the sense that all internal points belong to $\Pi(w)$. The prefix normal forms $\PNF_{a}(w)$ and $\PNF_{b}(w)$ are obtained by connecting the upper and the lower points of the region, respectively.

\begin{figure}
\begin{center}
  \includegraphics[height=35mm]{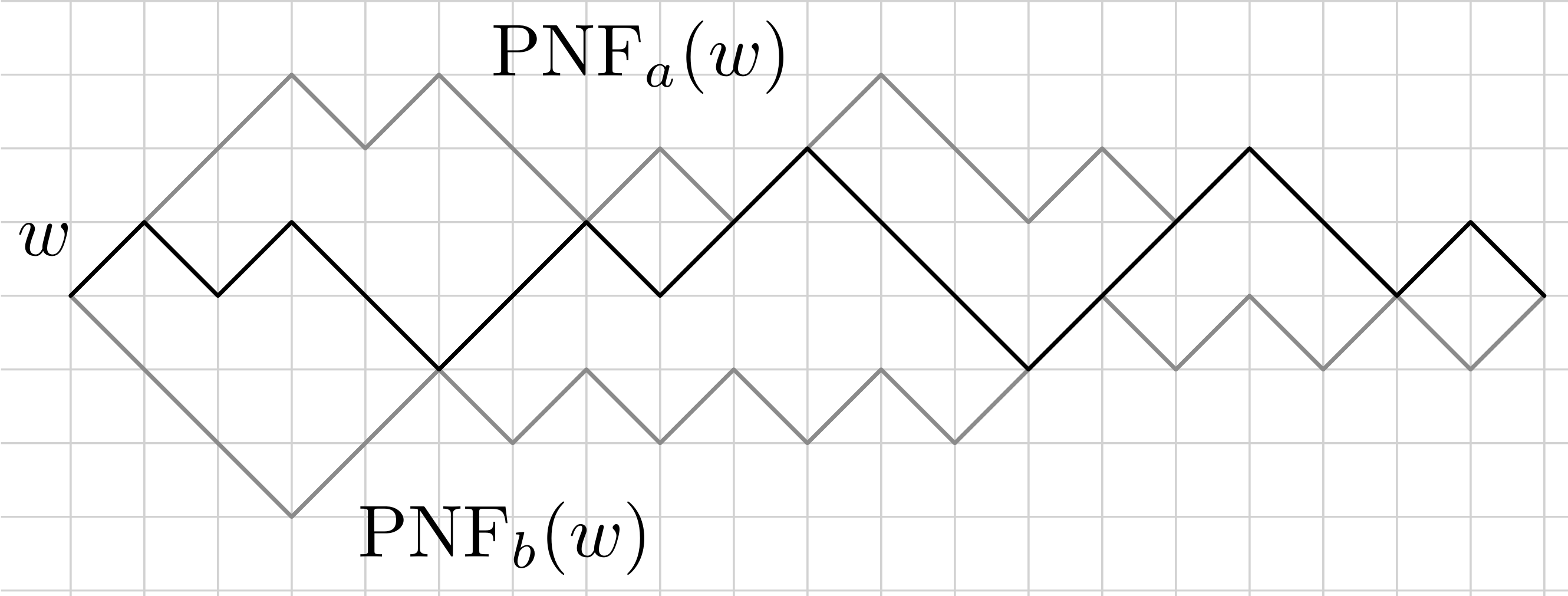}
  \caption{The word $w=ababbaabaabbbaaabbab$, its prefix normal forms $\PNF_{a}(w)=aaababbabaabbababbab$ and $\PNF_{b}(w)=bbbaababababaabababa$, and the region delineated by $\Pi(w)$, the Parikh set of $w$.}\label{fig:esempio}
\end{center}
\end{figure}

\section{The Language of Prefix Normal Words}\label{sec:La}

In this section, we take a closer look at those words which are in prefix normal form, which we refer to as \emph{prefix normal words}. For simplicity of exposition, from now on we only refer to prefix normality with respect to the letter $a$.

\begin{definition}
A \emph{prefix normal word} is a word $w\in \Sigma^*$ such that for every $0\le k\le |w|$, $\fmax_{a}(w,k)=\pr_{a}(w,k)$. That is, a word such that $w=\PNF_{a}(w)$. 
We denote by $L_a\subset \Sigma^*$ the language of prefix normal words.
\end{definition}

The following proposition gives some characterizations of prefix normal words. Recall that $\pr_a(w,i)  = |w_1\cdots w_i|_a$ is the number of $a$'s in the prefix of length $i$, and $\pos_a(w,i) = \min\{ k : |w_1\cdots w_k|_a = i\}$ is the position of the $i$'th $a$. When no confusion can arise, we write simply $\pr_a(i)$ and $ \pos_a(i)$. In particular, we have $\pr_{a}( \pos_a(i))=i$ and $\pos_a(P_a(i)) \leq i$.

\begin{proposition}\label{prop:char}
Let $w\in \Sigma^{*}$. The following properties are equivalent:

\begin{enumerate}
\item $w$ is a prefix normal word;
\item $\forall i,j$ where $0\le i\le j \le |w|$, we have $ \pr_{a}(j) -  \pr_{a}(i) \leq  \pr_{a}(j-i)$;
\item $\forall v\in \Fact(w)$ such that $|v|_{a}=i$, we have $|v|\ge  \pos_a(i)$;
\item $\forall i,j$ such that $i+j-1 \le |w|_{a}$, we have $ \pos_a(i) +  \pos_a(j) -1 \leq  \pos_a(i+j-1)$. 
\end{enumerate}
\end{proposition}

\begin{proof}

(1) $\Rightarrow$ (2). Follows from Lemma~\ref{lemma:Fa}, since $\pr_a(w) = \fmax_a(w)$. 

(2) $\Rightarrow$ (3). Assume otherwise. Then there exists $v\in \Fact(w)$ s.t.\ $|v| <  \pos_a(k)$, where $k=|v|_a$. Let $v = w_{i+1}\cdots w_j$, thus $j-i=k$. Then $\pr_a(j) - \pr_a(i) = k$. But $\pr_a(j-i) = \pr_a(|v|) \leq k-1 < k = \pr_a(j) - \pr_a(i)$, a contradiction.

(3) $\Rightarrow$ (4). Again assume that the claim does not hold. Then there are $i,j$ s.t.\ $ \pos_a(i+j-1) <  \pos_a(i) +  \pos_a(j) -1$. Let $k= \pos_a(j)$ and $l =  \pos_a(i+j-1)$ and define $v = w_k\cdots w_l$. Then $v$ has $i$ many $a$'s. But $|v| =  \pos_a(i+j-1) -  \pos_a(j) + 1 <  \pos_a(i) +  \pos_a(j) - 1 -  \pos_a(j) + 1 =  \pos_a(i)$, in contradiction to (3).

(4) $\Rightarrow$ (1). Let $v\in \Fact(w)$, $|v|_a = i$. We have to show that $\pr_a(|v|) \geq i$. This is equivalent to showing that $ \pos_a(i) \leq |v|$. Let $v = w_{l+1} \cdots w_r$, thus $\pr_a(r) - \pr_a(l) = i$. Let $j = \pr_a(l)+1$, thus the first $a$ in $v$ is the $j$'th $a$ of $w$. Note that we have $l<  \pos_a(j)$ and  $r\geq  \pos_a(i+j-1)$. By the assumption, we have $ \pos_a(i) \leq  \pos_a(i+j-1) -  \pos_a(j) + 1 \leq r-l = |v|$. \qed
\end{proof}

We now give some simple facts about the language $L_a$.

\begin{proposition}\label{prop:La_properties} Let $L_a$ be the language of prefix normal words. \hfill \phantom{.}
\begin{enumerate}
\item $L_{a}$ is prefix-closed, that is, any prefix of a word in $L_{a}$ is a word in $L_{a}$.
\item If $w\in L_{a}$, then any word of the form $a^{k}w$ or $wb^{k}$, $k\ge 0$, also belongs to $L_a$.
\item Let $|w|_a <3$. Then $w\in L_{a}$ iff either $w=b^n$ for some $n\ge 0$ or the first letter of $w$ is $a$.
\item Let $w\in \Sigma^{*}$. Then there exist infinitely many $v\in \Sigma^{*}$ such that $vw\in L_{a}$.
\end{enumerate}
\end{proposition}

\begin{proof}
The claims {\em 1., 2., 3.} follow easily from the definition. For {\em 4.},  note that for any $n\ge |w|$, the word $a^{n}w$ belongs to $L_{a}$. \qed
\end{proof}

We now deal with the question of how a prefix normal word can be extended to the right into another prefix normal word.

\begin{lemma}\label{lemma:test}
 Let $w\in L_{a}$. Then $wa\in L_{a}$ if and only if for every $0\le k<|w|$ the suffix of $w$ of length $k$ has less $a$'s than the prefix of $w$ of length $k+1$.
\end{lemma}

\begin{proof}
Suppose $wa\in L_{a}$. Fix $k$ and let $va$ be the suffix of $wa$ of length $k+1$. By definition of $L_{a}$ one has $|va|_{a}\le \pr_{a}(k+1)$, and therefore $|v|_{a}< \pr_{a}(k+1)$.

Conversely, let $v$ be the suffix of $w$ of length $k$. Since $w\in L_{a}$ one has $|v|_{a}\le \pr_{a}(k)$. We cannot have $|v|_{a}= \pr_{a}(k)$ and $w_{k+1}=b$ since by hypothesis we must have $|v|_{a}<\pr_{a}(k+1)$. Thus either $|v|_{a}< \pr_{a}(k)$ or $w_{k+1}=b$. In both cases we have then $|va|_{a}\le \pr_{a}(k+1)$. Since no suffix of $wa$ has more $a$'s than the prefix of $wa$ of the same length, and since $w\in L_{a}$, it follows that $wa\in L_{a}$. \qed
\end{proof}

We close this section by proving that $L_a$ is not context-free. Our proof is an easy modification of the proof that
Berstel and Boasson gave for the fact that the language of binary Lyndon words is not context-free \cite{BeBo97}.

\begin{theorem}\label{teor:CF}
$L_{a}$ is not context-free.
\end{theorem}

\begin{proof}
Recall that Ogden's iteration lemma (see e.g.~\cite{BerstelB90}) states that, for every context-free language $L$ there exists an integer $N$ such that, for any word $w\in L$ and for any choice of at least $N$ distinguished positions in $w$, there exists a factorization $w=xuyvz$ such that 

\begin{enumerate}
 \item either $x,u,y$ each contain at least one distinguished position, or $y,v,z$ each contain at least one distinguished position;
 \item the word $uyv$ contains at most $N$ distinguished positions;
 \item for any $n\ge 0$, the word $xu^{n}yv^{n}z$ is in $L$.
 \end{enumerate}
 
Now, assume that the language $L_{a}$ is context-free, and consider the word $w=a^{N+1}ba^{N}ba^{N+1}$ where $N$ is the constant of  Ogden's Lemma. It is easy to see that $w\in L_{a}$. Distinguish the central run of $N$ letters $a$. We claim that for every factorization $w=xuyvz$ of $w$, pumping $u$ and $v$ eventually results in a word in which the first run of $a$'s is not the longest one. Such a word cannot belong to $L_{a}$.
 
If $x,u,y$ each contain at least one distinguished position, then $u$ is non-empty and it is contained in the central run of $a$'s. Now observe that every word obtained by pumping $u$ and $v$ is prefixed by $a^{N+1}b$. Pumping $u$ and $v$, we then get a word $a^{N+1}ba^ms$, for some word $s$, where $m>N+1$. This word is not in $L_a$.
 
Suppose now that $y,v,z$ each contain at least one distinguished position. Then $v$ is non-empty and it is contained in the central run of  $a$'s. 

If $u$ is contained in the first run of $a$'s and it is non-empty, then, pumping \textit{down}, one gets a word of the form  $a^{k}ba^{m}ba^{N+1}$ with $k\le N$. This word is not in $L_{a}$. In all other cases ($u$ is contained in the second run of $a$'s and it is non-empty, or $u=\epsilon$, or $u$ contains the first $b$ of $w$), every word obtained by pumping $u$ and $v$ is prefixed by $a^{N+1}b$. Again, pumping $u$ and $v$, we obtain a word in which the first run of $a$'s is not the longest one.  \qed
\end{proof}

\section{Prefix Normal Words vs.\ Lyndon Words}\label{sec:preneck} 

In this section, we explore the relationship between the language $L_{a}$ of prefix normal words and some known classes of words defined by means of lexicographic properties.

A \emph{Lyndon word} is a word which is lexicographically (strictly) smaller than any of its proper non-empty suffixes. Equivalently, $w$ is a Lyndon word if it is the (strictly) smallest, in the lexicographic order, among its conjugates, i.e.,\ for any factorization $w=uv$, with $u,v$ non-empty words, one has that the word $vu$ is lexicographically greater than $w$ \cite{LothaireAlg}. Note that, by definition, a Lyndon word is primitive, i.e.,\ it cannot be written as $w=u^{k}$ for a $u\in \Sigma^{*}$ and $k>1$. Let us denote by $Lyn$ the set of Lyndon words over $\Sigma$. One has that $Lyn\not\subseteq  L_a$ and $L_a \not\subseteq Lyn$. For example, the word $w=abab$ belongs to $L_a$ but is not a Lyndon word since it is not primitive. An example of Lyndon word which is not in prefix normal form is $w=aabbabaabbb$.

A power of a Lyndon word is called a \emph{prime word} \cite{Knuth42} or \emph{necklace} (see \cite{BePe07} for more details and references on this definition). 

Let us denote by $ \PL$ the set of prefixes of powers of Lyndon words, also called sesquipowers (or fractional powers) of Lyndon words \cite{ChHaPe04}, or \emph{preprime words} \cite{Knuth42}, or also \emph{pre-necklaces}\cite{Ruskey92}. It is easy to see that $ \PL$ is in fact the set of prefixes of Lyndon words plus the powers of the letter $b$.

The next proposition shows that any prefix normal word different form a power of the letter $b$ is a prefix of a Lyndon word.

\begin{proposition}\label{prop:preflin}
Let $w\in L_{a}$ with $|w|_{a}>0$. Then the word $wb^{|w|}$ is a Lyndon word.
\end{proposition}

\begin{proof}
We have to prove that any non-empty suffix of $wb^{|w|}$ is greater than $wb^{|w|}$. Suppose by contradiction that there exists a non-empty suffix $v$ of $wb^{|w|}$ that is smaller than $wb^{|w|}$, and let $u$ be the longest common prefix between $v$ and $wb^{|w|}$. This implies that $u$ is followed by different letters when it appears as prefix of $v$ and as prefix of $wb^{|w|}$. Since we supposed that $v$ is smaller than $wb^{|w|}$, we conclude that $ub$ is prefix of $wb^{|w|}$ and $ua$ is prefix of $v$. Since $ua$ is a factor of $wb^{|w|}$ ending with $a$, $ua$ must be a factor of $w$ and therefore $ub$ is a prefix of $w$. Thus the factor $ua$ of $w$ has one more $a$ than the prefix $ub$ of $w$, contradicting the fact that $w$ is a prefix normal word. \qed
\end{proof}

We can now state the following result.

\begin{theorem}\label{theor:PL}
Every prefix normal word is a pre-necklace. That is, $L_{a}\subset  \PL$.
\end{theorem}

\begin{proof}
 If $w$ is of the form $b^{n}$, $n\ge 1$, then $w$ is a power of the Lyndon word $b$. Otherwise, $w$ contains at least one $a$ and the claim follows by Proposition \ref{prop:preflin}. \qed
\end{proof}

The languages $L_{a}$  and $ \PL$, however, do not coincide. The shortest word in $ \PL$ that does not belong to $L_{a}$ is $w=aabbabaa$. 
Below we give the table of the number of words in $L_{a}$ of each length, up to 16, compared with that of pre-necklaces. This latter sequence is listed in Neil Sloane's On-Line Encyclopedia of Integer Sequences~\cite{sloane}.

\begin{table}[ht]
\centering  
\begin{small}
\begin{raggedright}

\begin{tabular}{c *{30}{@{\hspace{2.1mm}}l}}
 $n$    & 1\hspace{1ex} & 2\hspace{1ex} & 3\hspace{1ex} & 4\hspace{1ex} & 5\hspace{1ex} & 6\hspace{1ex} & 7\hspace{1ex} &
8\hspace{1ex} & 9\hspace{1ex} & 10 & 11 & 12 & 13 & 14 & 15 & 16 \\
\hline \rule[-6pt]{0pt}{22pt}
$L_{a}\cap \Sigma^{n}$ & 2 & 3 & 5 & 8 & 14 & 23 & 41 & 70 & 125 & 218 & 395 & 697 & 1273 & 2279 & 4185 & 7568 \\
$ \PL\cap \Sigma^{n}$ & 2 & 3 & 5 & 8 & 14 & 23 & 41 & 71 & 127 & 226 & 412 & 747 & 1377 & 2538 & 4720 & 8800  \\
\hline \rule[-6pt]{0pt}{22pt}
\end{tabular}
\end{raggedright}\caption{\label{table:cardofL}The number of words in $L_{a}$ and in $\PL$ for each length up to 16.}
\end{small}
\end{table}

\section{The Prefix Normal Equivalence}\label{sec:equiv}

The prefix normal form $\PNF_{a}$  induces an equivalence relation on $\Sigma^*$, namely $u \equiv_{\PNF_{a}} v$ if and only if $\PNF_{a}(u) = \PNF_{a}(v)$. In Table \ref{table:classes4}, we give all prefix normal words of length $4$, and their equivalence classes.

\begin{table}[ht]
\centering \begin{small}
\begin{raggedright}

\begin{tabular}{l *{3}{@{\hspace{6mm}}c}}
$\PNF_{a}$   & class  & card.\\
\hline \rule[-2pt]{0pt}{3pt}\\
$aaaa$ & \{$aaaa$\} & 1\\
$aaab$ & \{$aaab$, $baaa$\}& 2\\
$aaba$ & \{$aaba$, $abaa$\}& 2\\
$aabb$ & \{$aabb$, $baab$, $bbaa$\}& 3\\
$abab$ & \{$abab$, $baba$\}& 2\\
$abba$ & \{$abba$\}& 1\\
$abbb$ & \{$abbb$, $babb$, $bbab$, $bbba$\}& 4\\
$bbbb$ & \{$bbbb$\}& 1\\
\hline \vspace{4mm}
\end{tabular}
\end{raggedright}\caption{The classes of words of length $4$ having the same prefix normal form.\label{table:classes4}}
\end{small}
\end{table}

An interesting question is how to characterize two words that have the same prefix normal form. 
The classes of this equivalence do not seem to follow regular patterns. For example, the words $aabababa$, $aabbaaba$, $abaababa$, $abaabbaa$, $ababaaba$, $abababaa$ all have the same prefix normal form $aabababa$, so that no simple statement about the lengths of the runs of the two letters seems to provide a characterization of the classes. This example also shows that the prefix normal form of a word $w$ is in general more complicated that just a rotation of $w$ or of its reversal (which is in fact the case for small lengths). 

Recall that  for word length $n\leq 16$, we listed the number of equivalence classes in Table~\ref{table:cardofL}. 
The sizes of the equivalence classes seem to exhibit an irregular behaviour. We report in Table~\ref{table:classes8}, for each of the $70$ equivalence classes for words of length $8$, the prefix normal form and the number of words in the class. Furthermore, we report the cardinality of the largest class of words for each length up to $16$ (Table~\ref{table:classes}). 

\begin{table}[ht]
\centering \begin{small}
\begin{raggedright}

\begin{tabular}{l *{1}{@{\hspace{2mm}}c@{\hspace{6mm}}} l *{1}{@{\hspace{2mm}}c@{\hspace{6mm}}} l *{1}{@{\hspace{2mm}}c@{\hspace{6mm}}} l *{1}{@{\hspace{2mm}}c}}
$\PNF_{a}$  & card. & $\PNF_{a}$  & card. & $\PNF_{a}$  & card. & $\PNF_{a}$  & card.\\
\hline \rule[-2pt]{0pt}{3pt}\\

$aaaaaaaa$ & 1 & $aaabaabb$ & 6 & $aabababa$ & 6 & $abababba$ & 2 \\
$aaaaaaab$ & 2 & $aaababaa$ & 2 & $aabababb$ & 9 & $abababbb$ & 4 \\
$aaaaaaba$ & 2 & $aaababab$ & 6 & $aababbaa$ & 2 & $ababbaba$ & 1 \\
$aaaaaabb$ & 3 & $aaababba$ & 4 & $aababbab$ & 8 & $ababbabb$ & 6 \\
$aaaaabaa$ & 2 & $aaababbb$ & 8 & $aababbba$ & 4 & $ababbbab$ & 4 \\
$aaaaabab$ & 4 & $aaabbaaa$ & 1 & $aababbbb$ & 10 & $ababbbba$ & 2 \\
$aaaaabba$ & 2 & $aaabbaab$ & 4 & $aabbaabb$ & 3 & $ababbbbb$ & 6 \\
$aaaaabbb$ & 4 & $aaabbaba$ & 2 & $aabbabab$ & 4 & $abbabbab$ & 2 \\
$aaaabaaa$ & 2 & $aaabbabb$ & 6 & $aabbabba$ & 3 & $abbabbba$ & 2 \\
$aaaabaab$ & 4 & $aaabbbaa$ & 2 & $aabbabbb$ & 8 & $abbabbbb$ & 5 \\
$aaaababa$ & 3 & $aaabbbab$ & 4 & $aabbbaab$ & 2 & $abbbabbb$ & 4 \\
$aaaababb$ & 6 & $aaabbbba$ & 2 & $aabbbaba$ & 2 & $abbbbabb$ & 3 \\
$aaaabbaa$ & 2 & $aaabbbbb$ & 6 & $aabbbabb$ & 6 & $abbbbbab$ & 2 \\
$aaaabbab$ & 4 & $aabaabaa$ & 1 & $aabbbbaa$ & 1 & $abbbbbba$ & 1 \\
$aaaabbba$ & 2 & $aabaabab$ & 4 & $aabbbbab$ & 4 & $abbbbbbb$ & 8 \\
$aaaabbbb$ & 5 & $aabaabba$ & 2 & $aabbbbba$ & 2 & $bbbbbbbb$ & 1 \\
$aaabaaab$ & 2 & $aabaabbb$ & 4 & $aabbbbbb$ & 7 & \\
$aaabaaba$ & 4 & $aababaab$ & 2 & $abababab$ & 2 & \\

\hline \vspace{4mm}
\end{tabular}
\end{raggedright}\caption{The cardinalities of the $70$ classes of words of length $8$ having the same prefix normal form. There are $7$ classes of length $1$, $24$ classes of length $2$, $5$ classes of length $3$, $16$ classes of length $4$, $2$ classes of length $5$, $9$ classes of length $6$, $1$ class of length $7$, $4$ classes of length $8$, $1$ class of length $9$ and $1$ class of length $10$.\label{table:classes8}}
\end{small}
\end{table}

\begin{table}[ht]
\centering  
\begin{small}
\begin{raggedright}
\begin{tabular}{c *{30}{@{\hspace{3.0mm}}l}}
 $n$    & 1\hspace{1ex} & 2\hspace{1ex} & 3\hspace{1ex} & 4\hspace{1ex} & 5\hspace{1ex} & 6\hspace{1ex} & 7\hspace{1ex} &
8\hspace{1ex} & 9\hspace{1ex} & 10 & 11 & 12 & 13 & 14 & 15 & 16 \\
\hline \rule[-6pt]{0pt}{22pt}
$\max |[w]|$ & 1 & 2 & 3 & 4 & 5 & 6 & 8 & 10 & 12 & 18 & 24 & 30 & 40 & 60 & 80 & 111 \\
\hline \rule[-6pt]{0pt}{22pt}
\end{tabular}
\end{raggedright}\caption{The maximum cardinality of a class of words having the same prefix normal form.\label{table:classes}}
\end{small}
\end{table}

\section{Conclusion and Open Problems}\label{sec:concl}

In this paper, we introduced the prefix normal form of a binary word. This construction arises in the context of indexing for jumbled pattern matching and provides a characterization of the set of factor Parikh vectors of a binary word. We then investigated the language $L_{a}$ of words which are in prefix normal form (w.r.t.\ the letter $a$).

Many open problems remain. Among these, the questions regarding the prefix normal equivalence were explored in Section~\ref{sec:equiv}: How can we characterize two words that have the same prefix normal form? Can we say anything about the number or the size of equivalence classes for a given word length?

Although we showed that the language $L_{a}$ is strictly contained in the language of pre-necklaces (prefixes of powers of Lyndon words), we were not able to find a formula for enumerating prefix normal words. A possible direction for attacking this problem would be finding a characterization of those pre-necklaces which are not prefix normal. Indeed, an enumerative formula for the pre-necklaces is known \cite{sloane}. 

Another open problem is to find an algorithm for testing whether a word is in prefix normal form. The best offline algorithms at the moment are the ones for computing an index for the jumbled pattern matching problem, and thus the prefix normal form~\cite{BuCiFiLi10a,MoRa10,MoRaJDA}; these have running time $O(n^2/\log n)$ for a word of length $n$, or  $O(n^2/\log^2 n)$ in the word-RAM model. However, testing may be easier than actually computing the PNF. Note also that Lemma~\ref{lemma:test} gives us an online testing algorithm, with time complexity $O(n^{2})$. Another, similar, online testing algorithm is provided by condition {\em 4.} of Proposition~\ref{prop:char}, with running time $O(|w|_a^2)$, which is, of course, again $O(n^2)$ in general.

\subsubsection*{Acknowledgements.}

We would like to thank Bill Smyth for interesting discussions on prefix normal words. We are also grateful to an anonymous referee who helped us improve the proof of Theorem 3.

\bibliographystyle{abbrv}
\bibliography{pnf}
\end{document}